\documentclass[letterpaper]{article}
\usepackage{uai2020}
\usepackage[margin=1in]{geometry}

\usepackage{times}

\usepackage[round,authoryear]{natbib}

\usepackage[utf8]{inputenc}
\usepackage{hyperref}      
\usepackage{url}       
\usepackage{comment}    
\usepackage{booktabs}      
\usepackage{amsfonts}      
\usepackage{nicefrac}     
\usepackage{microtype}    

\usepackage{graphicx}
\usepackage{caption}
\usepackage{subcaption}

\usepackage{tikz}
\usepgflibrary{arrows}
\usetikzlibrary{shapes}
\usetikzlibrary{arrows.meta}
\usetikzlibrary{positioning}

\usepackage{algorithm2e}
\usepackage{amsmath}
\usepackage{amssymb}
\usepackage{amsthm}
\usepackage{dsfont}
\usepackage{caption}

\theoremstyle{definition}
\newtheorem{thm}{Theorem}

\newtheorem{ass}[thm]{Assumption}
\newtheorem{cor}[thm]{Corollary}
\newtheorem{defn}[thm]{Definition}

\newtheorem{prop}[thm]{Proposition}

\DeclareMathOperator{\Pois}{Pois}

\newcommand{\G}{\mathcal{G}}
\newcommand{\I}{\mathcal{I}}
\newcommand{\ind}[3]{\langle #1, #2 \mid #3 \rangle}

\newcommand{\musep}[3]{#1 \perp_\mu #2 \mid #3}
\newcommand{\musepG}[4]{#1 \perp_\mu #2 \mid #3\ [#4]}

\newcommand{\an}{\text{an}}

\newcommand{\E}{\text{E}}
\newcommand{\pa}{\text{pa}}
\newcommand{\PR}{\text{P}}

\newcommand{\md}{\mathrm{d}}

\newenvironment{subalgorithm}[1][htb]
{
	\begin{algorithm}[#1]%
	}{\end{algorithm}}

\title{Causal screening in dynamical systems}

\author{%
	Søren Wengel Mogensen \\
	Department of Mathematical Sciences\\
	University of Copenhagen\\
	Copenhagen, Denmark \\
	\texttt{swengel@math.ku.dk} \\
}

\begin{document}

\maketitle

\begin{abstract}
	Many classical algorithms output graphical representations of causal 
	structures by testing conditional independence among a set of random 
	variables. In dynamical systems, local independence can be used analogously 
	as a testable implication of the underlying data-generating process. 
	We suggest some inexpensive methods for causal screening which provide 
	output with a sound causal interpretation under the assumption of ancestral 
	faithfulness. The popular model class of linear Hawkes processes is used 
	to provide an example of a dynamical causal model. 
	We argue 
	that for sparse causal graphs the output will often be close to complete. 
	We give 
	examples of this framework and apply it to a challenging biological system.
\end{abstract}

\section{INTRODUCTION}

Constraint-based causal learning is computationally and statistically 
challenging. There is a large literature on learning structures that are 
represented by directed acyclic graphs (DAGs) or margina\-li\-za\-tions thereof
(see \cite{handbookGraphical2019} for references). The fast causal 
inference algorithm \citep[FCI,][]{spirtes1993} provides in a certain sense 
maximally informative 
output \citep{zhang2008}, but 
at the cost of 
using a large number of conditional independence tests \citep{colombo2012}. To 
reduce 
the 
computational cost, other methods provide output which has a 
sound causal interpretation, but may be less informative. Among these are the 
anytime FCI 
\citep{spirtesAnytime2001} and 
RFCI \citep{colombo2012}. A recent algorithm, ancestral causal inference 
\citep[ACI,][]{magliacaneACI2016}, aims to learn only the directed part of 
the 
underlying graphical 
structure which allows for a sound causal interpretation even though some 
information is lost. 

In this paper, we describe some simple methods for learning causal structure in 
dynamical systems represented by stochastic processes. Many authors have 
described frameworks and algorithms for learning structure in 
systems of time series, ordinary differential 
equations, stochastic differential equations, and point processes. 
However, most of these methods do not have a clear causal interpretation when 
the 
observed processes are part of a larger system and most of the current 
literature is either non-causal in nature, or requires that there 
are no unobserved processes.

Analogously to testing conditional independence when learning DAGs, one can 
use 
tests of local independence in the case of dynamical systems. 
\cite{eichler2013}, \cite{meek2014}, and \cite{mogensenUAI2018} propose 
algorithms for 
learning graphs that represent local independence structures. 
We show empirically that we can recover features of their graphical learning 
target using 
considerably 
fewer tests of local independence. First, we suggest a learning 
target which is easier to learn, though still conveys 
useful causal information, analogously to ACI \citep{magliacaneACI2016}. 
Second, 
the proposed algorithm is only 
guaranteed to 
provide a supergraph of the learning target and this also reduces the number of 
local 
independence tests drastically. A central point is that our proposed methods 
retain a  causal interpretation in the sense that absent edges in the output 
correspond to implausible causal connections.

\cite{meek2014} suggests learning a directed graph to 
represent a causal dynamical system and 
gives 
a learning 
algorithm which we will describe as a {\it simple screening algorithm} (Section 
\ref{ssec:simpScreen}). We show 
that this algorithm can be given a sound interpretation under a weaker 
faithfulness assumption than that of \cite{meek2014}. We also 
provide a simple interpretation of the output of this algorithm and we show 
that similar screening algorithms can give comparable results using 
considerably 
fewer tests of local independence.

All proofs are provided in the supplementary material.

\begin{figure*}
	\begin{subfigure}{.48\linewidth}
		\centering
		\includegraphics[width=1\linewidth]{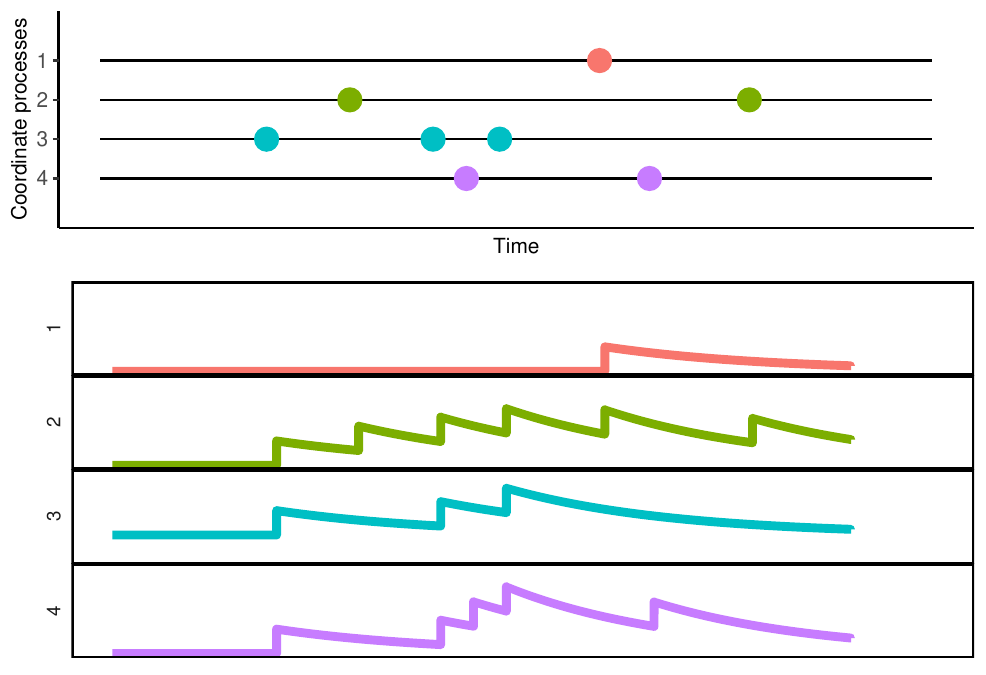}
		\caption{\label{fig:data} Top: Example data from a four-dimensional 
			Hawkes process. Bottom: The corresponding intensities. The time 
			axis is aligned between the two plots.}
	\end{subfigure} \hfill
	\begin{subfigure}{.48\linewidth}
		\centering
		\begin{tikzpicture}[scale=0.7]
		\tikzset{vertex/.style = {shape=circle,draw,minimum size=1.5em, inner 
				sep = 0pt}}
		\tikzset{vertexFac/.style = {shape=rectangle,draw,minimum size=1.5em, 
				inner sep = 0pt, gray}}
		\tikzset{edge/.style = {->,> = latex', thick}}
		\tikzset{edgebi/.style = {<->,> = latex', thick}}
		\tikzset{every loop/.style={min distance=8mm, looseness=5}}
		
		\node[vertex] (a) at  (0,0) {1};
		\node[vertex] (b) at  (2,0) {2};
		\node[vertex] (c) at  (0,-2) {3};
		\node[vertex] (d) at  (2,-2) {4};
		\draw[edge] (a) to (b);
		\draw[edge] (c) to (b);
		\draw[edge] (c) to (d);
		
		\draw[edge, loop left] (a) to (a);
		\draw[edge, loop right] (b) to (b);
		\draw[edge, loop left] (c) to (c);
		\draw[edge, loop right] (d) to (d);

		\node[vertex] (a1) at  (0+6,0) {1};
		\node[vertex] (b1) at  (2+6,0) {2};
		\node[vertex] (d1) at  (2+6,-2) {4};
		\draw[edge] (a1) to (b1);
		\draw[edge, bend left] (d1) to (b1);
		\draw[edge, bend left] (b1) to (d1);
		
		\draw[edge, loop left] (a1) to (a1);
		\draw[edge, loop right] (b1) to (b1);
		\draw[edge, loop right] (d1) to (d1);
		
		\end{tikzpicture}
		\vspace{.6cm}
		\caption{\label{fig:causal} Left: The causal graph (see Section 
			\ref{ssec:dynCauModel}) of a four-dimensional Hawkes process.
			Right: Learning output of standard approach (see Section 
			\ref{sec:hawkes}) when 3 is 
			unobserved. When 3 is unobserved, 2 is predictive of 4 and vice 
			versa 
			(heuristically, more events in process 2 indicate more events in 3 
			which in 
			turn 
			indicates 
			more events in 4). However, they are not causally connected and 
			using 
			local independence one can learn that 2 is not a parent of 4. This 
			is 
			important to predict what would happen under interventions in the 
			system as the right-hand graph indicates that an intervention on 2 
			would 
			change the distribution of 4 even though this is not the case as 
			$g^{\alpha 2} = 0$ for $\alpha\in \{1,3,4\}$.}
	\end{subfigure}
	\caption{\label{fig:fig1} Subfigure \ref{fig:data} shows data generated 
		from the system in \ref{fig:causal} (left). Until the first event all 
		intensities are constant (equal to $\mu_\alpha$ for the 
		$\alpha$-process). The first event occurs in process 3. We see that 
		$g^{23}$, $g^{33}$, and $g^{43}$ are different from zero as encoded by 
		the graph in \ref{fig:causal} (left). 
		Therefore the event makes the intensity processes of 2, 3, and 4 jump, 
		making new events in these processes more likely in the immediate 
		future (\ref{fig:data}, bottom).}
\end{figure*}
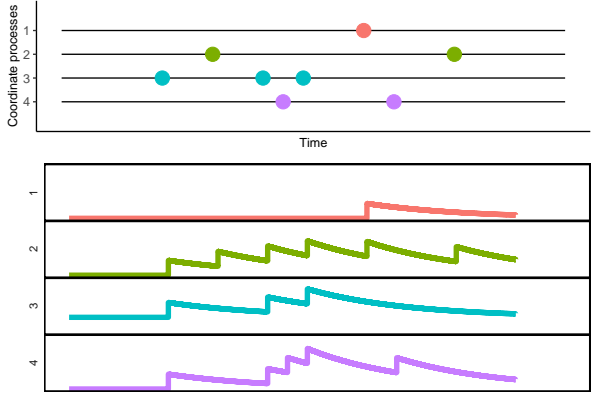
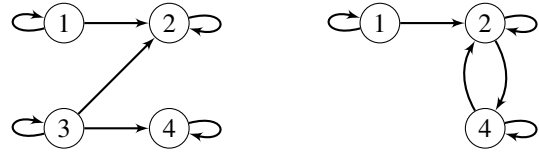

\section{HAWKES PROCESSES}
\label{sec:hawkes}

Local independence can be defined in a wide range of discrete-time and 
continous-time dynamical models (e.g., 
point 
processes \citep{didelez2000}, time 
series \citep{eichler2012}, and diffusions \citep{ mogensenUAI2018}. See also 
\cite{commenges2009}), and the algorithmic 
results we present apply to all these classes of models. However, the 
causal interpretation will differ between these model classes, and we will use 
the 
{\it linear 
	Hawkes 
	processes} to exemplify the framework. \cite{laubHawkes2015} give an 
	accessible introduction to this continuous-time
model class and \cite{linigerThesis}, \cite{bacry2015}, and \cite{daleyVere} 
provide more background. 
Hawkes 
processes have also 
been studied in the 
machine 
learning community in recent years \citep{zhou2013, zhou2013b, luo2015, xu2016, 
etesami2016, 
achab2017, tan2018, xu2018, trouleau2019}. It 
is 
important to note 
that these 
papers all consider the case of full observation, i.e., every coordinate 
process 
is observed. In causal systems that are not fully observed that assumption 
may lead to false conclusions (see Figure \ref{fig:causal}). Our work addresses 
the learning problem without 
the assumption of full observation, hence there can be unknown and unobserved 
confounding processes. 

On a filtered probability space, $(\Omega, 
\mathcal{F}, 
(\mathcal{F}_t), 
\PR)$, 
we consider an $n$-dimensional
multivariate point process, $X = (X^1, \ldots, X^n)$.  
$\mathcal{F}_t$ is a filtration, i.e., a nondecreasing family of 
$\sigma$-algebras, and it represents the information which is available at a 
specific 
time point. Each coordinate 
process $X^\alpha$ is described by a sequence of positive, stochastic event 
times 
$T_1^\alpha, T_2^\alpha, \ldots$ such that $T_j^\alpha > T_i^\alpha$ almost 
surely for $j 
>i$. We 
let $V 
= \{1,\ldots,n\}$. This can also be formulated in terms of a counting 
process, 
$N$, such that $N_s^\alpha = \sum_i \mathds{1}_{(T_i \leq s)} $, $\alpha\in 
V$. There exists so-called 
{\it intensity 
	processes}, $\lambda = (\lambda^1,\ldots,\lambda^n)$, such that 

\[
\lambda_t^\alpha = \lim_{h\rightarrow 0} \frac{1}{h} \PR(N_{t+h}^\alpha 
- N_t^\alpha = 1 \mid \mathcal{F}_t)
\]

and the intensity at time $t$ can therefore be thought of as describing 
the 
probability 
of a 
jump in the immediate future after time $t$ conditionally on the 
history until time $t$ as captured by the $\mathcal{F}_t$-filtration. 
In 
a linear Hawkes model, 
the intensity of the $\alpha$-process, $\alpha\in V$, is of the simple form

\begin{align*}
\lambda_t^\alpha &= \mu_\alpha + \sum_{\gamma \in V} \int_{0}^{t} 
g^{\alpha\gamma}(t-s)\ \md N_s^\gamma \\ & = \mu_\alpha + \sum_{\gamma\in V} 
\sum_{i:T_i^\gamma < t} 
g^{\alpha\gamma}(t-T_i^\gamma)
\end{align*}

\noindent where $\mu_\alpha \geq 0$ and the function $g^{\alpha\gamma} 
: 
\mathbb{R}_{+} 
\rightarrow \mathbb{R}$ is nonnegative for all  $\alpha,\gamma \in V$. From 
the above formula, we see that if $g^{\beta\alpha} = 0$, then the 
$\alpha$-process does not enter directly into the intensity of the 
$\beta$-process and we will formalize this 
observation in subsequent sections. The intensity processes determine how the 
Hawkes process evolves and if $g^{\beta\alpha} = 0$ then the 
$\alpha$-process does not directly influence the evolution of the 
$\beta$-process (it may of course have an indirect influence which is mediated 
by other processes). Figure \ref{fig:data} provides an example of data from a 
linear Hawkes process and an illustration of its intensity processes.

\begin{figure*}
	\begin{subfigure}{0.48\linewidth}
		\centering
		\begin{tikzpicture}[scale=0.7]
		\tikzset{vertex/.style = {shape=circle,draw,minimum size=1.5em, inner 
				sep = 0pt}}
		\tikzset{edge/.style = {->,> = latex', thick}}
		\tikzset{edgebi/.style = {<->,> = latex', thick}}
		\tikzset{every loop/.style={min distance=8mm, looseness=5}}
		\tikzset{vertexFac/.style = {shape=rectangle,draw,minimum size=1.5em, 
				inner sep = 0pt}}
		
		\node[vertex] (a) at  (-3,0) {$\alpha$};
		\node[vertexFac] (b) at  (-1.5,2) {$\beta$};
		\node[vertexFac] (c) at  (0,0) {$\gamma$};
		\node[vertex] (d) at  (1.5,2) {$\delta$};
		\node[vertex] (e) at  (3,0) {$\epsilon$};
		\node[vertexFac] (f) at  (0,4) {$\phi$};
		
		\draw[edge, loop left] (a) to (a);
		\draw[edge, loop left] (b) to (b);
		\draw[edge, loop right] (c) to (c);
		\draw[edge, loop right] (d) to (d);
		\draw[edge, loop right] (e) to (e);
		\draw[edge, loop left] (f) to (f);
		\draw[edge, bend left = 45] (f) to (e);
		
		\draw[edge] (a) to (b);
		\draw[edge, bend left = 20] (b) to (c);
		\draw[edge, bend left = 20] (c) to (b);
		\draw[edge] (b) to (d);
		\draw[edge] (d) to (c);
		\draw[edge] (d) to (e);
		\draw[edge] (f) to (b);
		\draw[edge] (f) to (d);
		
		\end{tikzpicture}
	\end{subfigure}
	\begin{subfigure}{0.48\linewidth}
		\centering
		\begin{tikzpicture}[scale=0.7]
		\tikzset{vertex/.style = {shape=circle,draw,minimum size=1.5em, inner 
				sep = 0pt}}
		\tikzset{vertexFac/.style = {shape=rectangle,draw,minimum size=1.5em, 
				inner sep = 0pt, gray}}
		\tikzset{edge/.style = {->,> = latex', thick}}
		\tikzset{edgebi/.style = {<->,> = latex', thick}}
		\tikzset{every loop/.style={min distance=8mm, looseness=5}}
		
		\node[vertex] (a) at  (-3,0) {$\alpha$};
		\node[vertexFac] (b) at  (-1.5,2) {$\beta$};
		\node[vertexFac] (c) at  (0,0) {$\gamma$};
		\node[vertex] (d) at  (1.5,2) {$\delta$};
		\node[vertex] (e) at  (3,0) {$\epsilon$};
		\node[vertexFac] (f) at  (0,4) {$\phi$};
		
		\draw[edge, loop left] (a) to (a);
		\draw[edge, loop above] (d) to (d);
		\draw[edge, loop right] (e) to (e);
		
		\draw[edge] (a) to (d);
		\draw[edge] (d) to (e);
		
		\end{tikzpicture}
	\end{subfigure}
	\caption{\label{fig:paG} Left: A causal graph on nodes $V = 
		\{\alpha,\beta,\gamma,\delta,\epsilon, \phi \}$. Right: The 
		corresponding 
		parent 
		graph on nodes $O = \{\alpha,\delta,\epsilon\}$. Note that causal 
		graphs 
		and parent graphs may contain cycles. The parent graph does not contain 
		information on the confounder process $\phi$ as it only encodes `causal 
		ancestors'. One can also {\it marginalize} the causal graph to obtain a 
		{\it directed mixed graph} from which one can read off the parent 
		graph (see the supplementary material).}
\end{figure*}
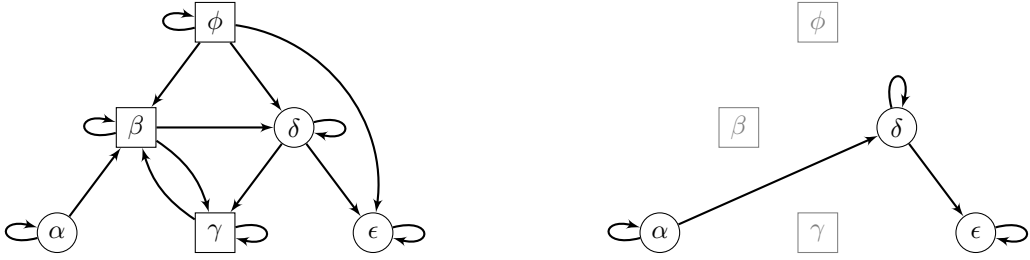

\subsection{A DYNAMICAL CAUSAL MODEL}
\label{ssec:dynCauModel}

We will in this section define what we mean by a {\it dynamical causal model} 
in the case of a linear Hawkes process and also 
define a graph $(V,E)$ which represents the causal structure of the model. The 
node set $V$ is the index set of the coordinate 
processes of the multivariate Hawkes process, thus identifying each node with a 
coordinate process. If we first consider the case where $X = (X_1,\ldots,X_n)$ 
is a multivariate 
random variable, it is common to define a {\it causal} model in terms of a DAG, 
$\mathcal{D}$,  
and a 
structural causal 
model \citep{pearl2009, petersElements2017} by assuming that there exists 
functions $f_i$ and 
error 
terms $\epsilon_i$ such that

\[
X_i = f_i(X_{\pa_\mathcal{D}(X_i)}, \epsilon_i)
\]

for $i = 1,\ldots,n$. The causal assumption amounts to assuming that the
functional relations are stable under interventions. This idea can
be transferred to dynamical 
systems (see also \cite{roysland2012, mogensenUAI2018}).  In the case of a 
linear Hawkes process as described above, we can consider intervening on the 
$\alpha$-process and 
force events to occur at the deterministic times 
$t_1, \ldots, t_k$, and at these times only. In this case, the causal 
assumption 
amounts to assuming 
that the distribution of the intervened system is
governed by the intensities

\begin{align*}
\lambda_t^\beta = \mu_\beta &+ \int_{0}^t g^{\beta\alpha}(t-s)\ \md 
\bar{N}_s^\alpha \\
&+ 
\sum_{\gamma\in V \setminus \{\alpha\}} \int_{0}^t g^{\beta\gamma}(t-s)\ 
\md N_s^\gamma 
\end{align*}

for all $\beta \in V\setminus \{\alpha\}$ and where $\bar{N}_t^\alpha = \sum_{i 
	= 1}^{k} \mathds{1}_{(t_i \leq t)}$. We will not go into a discussion of 
the 
existence 
of these interventional stochastic processes. The above is a {\it hard} 
intervention in the sense that the $\alpha$-process is fixed to be a 
deterministic function of time. Note that one could easily imagine 
other types of interventions such as {\it soft} interventions where the 
intervention process, $\alpha$, is not deterministic. One can also extend this 
to interventions on more than one process. It holds that $N_{t + 
h}^\beta - 
N_t^\beta \sim 
\Pois(\lambda_t^\beta \cdot h)$ in the limit $h\rightarrow 
0$, and we can think of this as a simulation scheme in which we generate the 
points in one small interval in accordance to some distribution depending on 
the history of the process. As such the intensity describes a structural causal 
model at infinitesimal time steps and the $g^{\alpha\beta}$-functions are in a 
causal model 
stable under interventions in the sense that they also describe how the 
intervention process $\bar{N}^\alpha$ enters into the intensity of the other 
processes.

We use the set of functions 
$\{g^{\beta\alpha}\}_{\alpha,\beta\in 
	V}$ to define the {\it causal graph} of the Hawkes process. A {\it graph} 
is a pair $(V,E)$ where $V$ is a set of nodes and $E$ is a set of edges 
between these nodes. We assume that we observe the Hawkes process in the 
time interval $J = [0,T]$, $T\in \mathbb{R}$. The causal graph has node 
set $V$ (the index set of 
the coordinate processes) and the edge $\alpha\rightarrow\beta$ is in the 
causal graph if and only if $g^{\beta\alpha}$ is not 
identically zero on $J$. We call this graph {\it 
	causal} as it is defined using $\{g^{\beta\alpha}\}_{\alpha,\beta \in V}$ 
	which 
is 
a set of mechanisms assumed stable under interventions, and this 
causal 
assumption is therefore analogous to that of a 
classical structural causal 
model as briefly introduced above.

\subsection{PARENT GRAPHS}

In principle, we would like to recover the causal graph, $\mathcal{D}$, 
using local independence tests. 
Often, we will only have partial observation of the dynamical 
system in the 
sense that we only observe the processes in $O\subsetneq V$. We will then aim 
to 
learn the {\it parent graph} of $\mathcal{D}$ on nodes $O$. 

\begin{defn}[Parent graph]
	Let $\mathcal{D} = (V,E)$ be a causal graph and let $O\subseteq V$. The 
	{\it parent graph} of 
	$\mathcal{D}$ on nodes $O$ is the graph $(O,F)$ such that for 
	$\alpha,\beta\in O$, the edge
	$\alpha\rightarrow \beta$ is in $F$ if and only if the edge 
	$\alpha\rightarrow\beta$ is in the causal graph or there is
	a path $\alpha 
	\rightarrow \delta_1 \rightarrow \ldots 
	\rightarrow \delta_k \rightarrow \beta$  in the causal graph such that 
	$\delta_1,\ldots,\delta_k \notin O$, for some $k>0$ . 
	\label{def:paG}
\end{defn}

We denote the parent 
graph of 
the causal graph by $\mathcal{P}_O(\G)$, or just $\mathcal{P}(\G)$ if the set 
$O$ 
used is clear from the context. In applications, a parent graph may provide 
answers to important 
questions as it tells us the causal relationships 
between the observed nodes. A similar idea was applied in DAG-based models by 
\cite{magliacaneACI2016}, though that paper describes an exact method and not a 
screening procedure. In large systems, it can easily be infeasible to 
learn the
complete independence structure of the observed system, and we propose instead 
to estimate the parent graph which can be done efficiently. In 
the supplementary material, we give another characterization of a parent graph. 
Figure 
\ref{fig:paG} contains an example of a causal graph and a corresponding parent 
graph.

\subsection{LOCAL INDEPENDENCE}

Local independence has been studied by several authors and in different classes 
of continuous-time models as well as in time series \citep{aalen1987, 
	didelez2000, didelez2008, eichler2010}. We give an abstract definition of 
local 
independence, following the exposition by \cite{mogensenUAI2018}.

\begin{defn}[Local independence]
	Let $X$ be a multivariate stochastic process and let $V$ be an index set of 
	its coordinate processes. Let $\mathcal{F}_t^D$ denote the complete and 
	right-continuous version of the $\sigma$-algebra 
	$\sigma(\{X_s^\alpha: s\leq t,\alpha \in D \})$, $D \subseteq V$. Let 
	$\lambda$ be a 
	multivariate stochastic 
	process (assumed to be integrable and c\`adl\`ag) such that its coordinate 
	processes are 
	indexed by $V$. For 
	$A,B,C\subseteq V$, we say that $X^B$ is $\lambda$-locally independent of 
	$X^A$ given $X^C$ (or simply $B$ is $\lambda$-locally independent of $A$ 
	given $C$) if the process
	
	\[
	t \mapsto \E(\lambda_t^\beta \mid \mathcal{F}_t^{C\cup A})
	\]
	
	has an $\mathcal{F}_t^C$-adapted version for all $\beta\in B$. We write 
	this as 
	$A\not\rightarrow_\lambda B \mid C$, or simply $A\not\rightarrow B \mid C$.
\end{defn}

In the case of Hawkes processes, the intensities will be used as the 
$\lambda$-processes in the above definition. \cite{didelez2000}, 
\cite{mogensenUAI2018}, and 
\cite{mogensen2018} provide technical details on the definition of local 
independence. Local independence can be 
thought of as a dynamical system analogue to the classical conditional 
independence. It is, however, asymmetric which means that $A\not\rightarrow 
B\mid C$ does not imply $B\not\rightarrow A\mid C$. This is a natural and 
desirable feature of an independence relation in a dynamical system as it helps 
us distinguish between the past and the present. It is important to note that 
by testing local independences we can obtain more information about the 
underlying parent graph than by simply  
assuming full observation and fitting a model to the observed data
(see Figure \ref{fig:causal}).

\subsubsection{Local Independence and the Causal Graph}

To make progress on the learning task, we will in this subsection describe the 
link between the local independence model and the causal graph.

\begin{defn}[Pairwise Markov property \citep{didelez2008}]
	We say that a local independence model satisfies the {\it pairwise 
		Markov 
		property} with respect to a directed graph, $\mathcal{D} = (V,E)$, if 
		the absence 
	of the edge $\alpha 
	\rightarrow\beta$ in $\mathcal{D}$ implies $\alpha 
	\not\rightarrow_\lambda \beta 
	\mid V\setminus \alpha$ for all $\alpha,\beta\in V$. 
\end{defn}

We will make 
the following technical assumption throughout the paper. In 
applications, the functions $g^{\alpha\beta}$ are often assumed to be of the 
below 
type (\cite{laubHawkes2015}).

\begin{ass}
	Assume that $N$ is a multivariate Hawkes process and that we observed $N$ 
	over the interval $J = [0,T]$ where $T > 0$. For all $\alpha,\beta\in V$, 
	the function $g^{\beta\alpha}: \mathbb{R}_+ \rightarrow \mathbb{R}$ is 
	continuous and $\mu_\alpha > 0$.
	\label{ass:gFunc}
\end{ass}

A version of the following result was also 
stated by \cite{eichlerHawkes2017} but no proof was given and we provide one in 
the supplementary material. If $\G_1 = (V,E_1)$ 
and $\G_2 = (V,E_2)$ are graphs, we say that $\G_1$ is a {\it proper subgraph} 
of $\G_2$ if $E_1 \subsetneq E_2$.

\begin{prop}
	The local independence model of a linear Hawkes process satisfies the 
	pairwise Markov property with respect to the causal graph of the process 
	and no proper subgraph of the causal graph has the property.
	\label{prop:pairCau}
\end{prop}

\section{GRAPH THEORY AND INDEPENDENCE MODELS}

A {\it graph} is a pair $(V,E)$ 
where $V$ is a finite set of nodes and $E$ a 
finite set of 
edges. We will use $\sim$ to denote a 
generic edge. Each edge is between a pair of nodes (not necessarily 
distinct), and for 
$\alpha,\beta\in V$, $e \in E$, we will write $\alpha\overset{e}{\sim} \beta$ 
to denote that the edge $e$ is between $\alpha$ and $\beta$. We will in 
particular consider the class of {\it 
	directed graphs} (DGs) where between each pair of nodes $\alpha,\beta\in 
	V$
one 
has a subset of the edges $\{\alpha\rightarrow\beta, \alpha\leftarrow\beta\}$, 
and we say that these edges are {\it directed}.

Let $\G_1 = (V,E_1)$ and $\G_2 = (V,E_2)$ be graphs. We say that $\G_2$ is a 
{\it 
	supergraph} of $\G_1$, and write $\G_1 \subseteq \G_2$, if $E_1\subseteq 
	E_2$. 
For a graph $\G = (V,E)$ such that $\alpha,\beta \in V$, we write 
$\alpha\rightarrow_\G \beta$ to indicate that the directed edge from $\alpha$ 
to 
$\beta$ is contained in the edge set $E$. In this case we say that $\alpha$ is 
a {\it parent} of $\beta$. We let $\pa_\G(\beta)$ denote the set of nodes in 
$V$ that are parents of $\beta$. We write 
$\alpha\not\rightarrow_\G \beta$ to indicate that the edge is {\it not} in $E$. 
Earlier work allowed loops, i.e., self-edges $\alpha\rightarrow\alpha$, 
to be either present or absent in the graph \citep{meek2014,mogensenUAI2018, 
mogensen2018}. We assume that all loops are present, though this is not an 
essential 
assumption.

A {\it walk} is a finite 
sequence of nodes, $\alpha_i \in V$, and edges, $e_i\in E$, $\langle \alpha_1, 
e_1, 
\alpha_2, \ldots ,\alpha_k, e_k, \alpha_{k+1} \rangle$ such that $e_i$ is 
between $\alpha_i$ and $\alpha_{i+1}$ for all $i = 1,\ldots,k$ and such that 
an 
orientation of each edge is known. We say that a walk is 
{\it nontrivial} if it contains at least one edge. A {\it 
	path} is a walk such that no node is repeated. A {\it 
	directed} path from $\alpha$ to $\beta$ is a path such that all edges are 
directed and point 
in the direction of $\beta$.

\begin{defn}[Trek, directed trek]
	A {\it trek} between $\alpha$ and $\beta$ is a (nontrivial) path $\langle 
	\alpha, 
	e_1,\ldots,e_k,\beta\rangle$ with no colliders \citep{foygelHalftrek2012}. 
	We say that a trek 
	between 
	$\alpha$ and $\beta$ is {\it directed} from $\alpha$ to $\beta$ if $e_k$ 
	has a head at $\beta$.
\end{defn}

We will formulate the following properties using a general {\it independence 
	model},  $\I$, on $V$. Let $\mathbb{P}(\cdot)$ denote the power set of some 
set. An independence model on $V$ is simply a subset of $\mathbb{P}(V) \times 
\mathbb{P}(V)\times \mathbb{P}(V)$ and can be thought of as a collection of 
independence statements that hold among the processes/variables indexed by $V$.
In subsequent 
sections, the independence models will be defined using the notion of local 
independence. In this case, for $A,B,C\subseteq V$, $A\not\rightarrow_\lambda 
B\mid C$ is equivalent to 
writing $\langle A,B\mid C\rangle \in \I$ in the abstract notation, and we use 
the two 
interchangeably. We do not require $\I$ to be symmetric, i.e., $\langle A,B\mid 
C\rangle \in \I$ does not imply $\langle B,A\mid C\rangle \in \I$. In the 
following, we also use $\mu$-separation which is a 
ternary 
relation and a dynamical model (and asymmetric) analogue to $d$-separation 
or $m$-separation. 

\begin{defn}[$\mu$-separation]
	Let $\G = (V,E)$ be a DMG, and let $\alpha,\beta\in V$ and $C\subseteq 
	V$. 
	We say that a (nontrivial) walk from $\alpha$ to $\beta$, $\langle 
	\alpha, 
	e_1,\ldots,e_k, \beta\rangle$, is $\mu$-connecting given $C$ if $\alpha 
	\notin C$, the edge $e_k$ has a head at $\beta$, every collider on the 
	walk
	is in $\an(C)$ and no noncollider is in $C$. Let $A,B,C \subseteq V$. 
	We 
	say that $B$ is $\mu$-separated from $A$ given $C$ if there is no 
	$\mu$-connecting walk from any $\alpha\in A$ to any $\beta\in B$ given 
	$C$. 
	In this case, we write 
	$\musep{A}{B}{C}$, or $\musepG{A}{B}{C}{\G}$ if we wish to emphasize 
	the 
	graph to
	which the statement relates.
\end{defn}

More graph-theoretical definitions and references are given in 
the supplementary material.

\begin{defn}[Global Markov property]
	We say that an independence model $\I$ satisfies the {\it global Markov 
		property} with respect to a DG, $\G = (V,E)$, if $\musepG{A}{B}{C}{\G}$ 
	implies $\langle A,B \mid C\rangle \in \I$ for all $A,B,C \subseteq V$.
\end{defn}

From Proposition \ref{prop:pairCau}, we know that the local independence model 
of a linear Hawkes process satisfies the pairwise Markov property with 
respect to its causal graph, and using the results in 
\cite{didelez2008} and \cite{mogensenUAI2018} it also satisfies the global 
Markov 
property 
with respect 
to this graph.

\begin{defn}[Faithfulness]
	We say that $\I$ is {\it faithful} with respect to a DG, $\G = (V,E)$, 
	if $\langle A,B \mid C\rangle \in \I$ implies $\musepG{A}{B}{C}{\G}$ for 
	all $A,B,C \subseteq V$.
	\label{def:faith}
\end{defn}

\section{NEW LEARNING ALGORITHMS}

In this section, we state a very general class of algorithms which 
is easily 
seen to 
provide sound causal learning and we describe some specific 
algorithms. We throughout assume that there is some underlying, true DG, 
$\mathcal{D}_0 = (V,E)$, 
describing the causal model and we wish to output 
$\mathcal{P}_O(\mathcal{D}_0)$.
However, this graph is not in general identifiable from the local independence 
model. In the supplementary material, we argue that for an equivalence class of 
parent graphs, there exists a unique member of the class which is a supergraph 
of all other members. Denote this unique graph by $\bar{\mathcal{D}}$. Our 
algorithms 
will output supergraphs of $\bar{\mathcal{D}}$, and the output will therefore 
also be supergraphs of the true
parent graph.

We assume that we are in the `oracle case', i.e., have access to a local 
independence oracle that provides the correct answers. We will say that an 
algorithm is 
{\it sound} if it in the oracle case outputs a 
supergraph of  $\bar{\mathcal{D}}$ and that it is {\it complete} if it outputs 
$\bar{\mathcal{D}}$. We let $\I^O$ denote the local independence model 
restricted to subsets of $O$, 
i.e., this is the observed part of the local independence model. We provide 
algorithms that are guaranteed to be sound, but only complete in particular 
cases. Naturally, one would wish for completeness as well. However, 
complete algorithms can easily be computationally infeasible whereas 
sound algorithms can be very inexpensive 
\citep[e.g., 
][]{mogensenUAI2018}. We think of these sound algorithms as {\it screening 
procedures} as they rule out some causal connections, but do not ensure 
completeness.

\subsection{ANCESTRAL FAITHFULNESS}

Under the faithfulness assumption, every local 
independence implies 
$\mu$-separation in the graph. We assume a weaker, but similar, property to 
show soundness. For learning marginalized 
DAGs, weaker types of 
faithfulness have also been explored, see \cite{zhangRobust2008, 
	zhalamaSATweakerAssump, 
	zhalamaDS2017}.

\begin{defn}[Ancestral faithfulness]
	Let $\I$ be an independence model and let
	$\mathcal{D}$ be a DG. We say that $\I$ satisfies {\it ancestral 
		faithfulness} with respect to $\mathcal{D}$ 
	if for every $\alpha,\beta\in V$ and $C\subseteq V\setminus \{\alpha\}$,  
	$\ind{\alpha}{\beta}{C} \in \I$ implies that there is no 
	$\mu$-connecting directed path from $\alpha$ 
	to 
	$\beta$ given $C$ in $\mathcal{D}$.
\end{defn}

Ancestral faithfulness is a strictly weaker 
requirement than faithfulness. We conjecture that local independence models of 
linear Hawkes processes
satisfy ancestral faithfulness with respect to their causal graphs. 
Heuristically, if there is a directed path from $\alpha$ to $\beta$ which is 
not blocked by any node in $C$, then information should flow from $\alpha$ to 
$\beta$, and this cannot be `cancelled out' by other paths in the graph as the 
linear Hawkes processes are self-excitatory, i.e., no process has a dampening 
effect on any process. This conjecture is supported by the so-called {\it 
Poisson cluster representation} of a linear Hawkes process (see 
\cite{jovanovic2015}).

\subsection{SIMPLE SCREENING ALGORITHMS}
\label{ssec:simpScreen}

As a first step in describing a causal screening algorithm, we will define a 
very general class of learning algorithms that simply test local independences 
and sequentially remove edges. It is easily seen 
that under the assumption of ancestral faithfulness every algorithm in this 
class gives sound learning in the oracle case. The {\it complete} DG on nodes 
$V$ is 
the DG with edge set $\{\alpha\rightarrow\beta \mid \alpha,\beta\in V\}$.

\begin{defn}[Simple screening algorithm]
	We say that a learning algorithm is a {\it simple screening algorithm} if 
	it starts from a complete DG on nodes $O$ and removes an edge $\alpha 
	\rightarrow \beta$ only if a conditioning set $C \subseteq O\setminus 
	\{\alpha\}$ has been found such that $\langle \alpha,\beta \mid C\rangle 
	\in \I^O$.
\end{defn}

The next results describe what can be learned from absent edges in the output 
of a simple screening algorithm.

\begin{prop}
	Assume that $\I$ satisfies ancestral faithfulness with respect 
	to 
	$\mathcal{D}_0 = (V,E)$. The 
	output 
	of any simple screening algorithm is sound in the oracle case.
	\label{prop:soundScreen}
\end{prop}

\begin{cor}
	Assume ancestral faithfulness of $\I$ with respect to $\mathcal{D}_0$ and 
	let $A,B,C\subseteq O$. If 
	every 
	directed path from $A$ 
	to $B$ goes through $C$ in the output graph of a simple screening 
	algorithm, then every directed path from 
	$A$ to $B$ goes through $C$ in $\mathcal{D}_0$. 
	\label{cor:soundPath}
\end{cor}

\begin{cor}
	If there is no directed path from $A$ to $B$ in the output graph, then 
	there is no directed path from $A$ to $B$ in $\mathcal{D}_0$. 
\end{cor}

\subsection{PARENT LEARNING}

In the previous section, it was shown that if edges are only removed when a 
separating 
set is found the output is sound under the assumption of ancestral 
faithfulness. In this section we give a specific algorithm. The key observation 
is that we can easily retrieve structural information from 
a rather small subset of local independence tests.

Let $\mathcal{D}^t$ denote the output from Subalgorithm \ref{subalgo:trekStep} 
(see below). 
The following result shows that under the 
assumption of faithfulness, $\alpha\rightarrow_{\mathcal{D}^t} \beta$ if and 
only if there is a directed trek from $\alpha$ to $\beta$ in $\mathcal{D}_0$.

\begin{prop}
	There is no directed trek from $\alpha$ to $\beta$ in 
	$\mathcal{D}_0$ if 
	and only if
	$\musepG{\alpha}{\beta}{\beta}{\mathcal{D}_0}$.
	\label{prop:trekSep}
\end{prop}

Note that above, $\beta$ in the conditioning set represents the $\beta$-past 
while the other $\beta$ represents the present of the $\beta$-process. While 
there is no distinction in the graph, this interpretation follows from the 
definition of local independence and the global Markov property. We will refer 
to running first Subalgorithm \ref{subalgo:trekStep} and then 
Subalgorithm \ref{subalgo:paStep} (using the output DG from the first as 
input to the second) as the causal screening (CS) algorithm. Intuitively, 
Subalgorithm \ref{subalgo:paStep} simply tests if a candidate set (the parent 
set) is a separating set 
and other candidate sets could be chosen.

\begin{prop}
	The CS algorithm is a simple screening algorithm.
\end{prop}

It is of course of interest to understand under what conditions the edge 
$\alpha\rightarrow\beta$ is guaranteed to be removed by the CS algorithm when 
it is 
not in the 
underlying target graph. In the supplementary material we state and prove 
a result describing one such condition.

	\begin{subalgorithm}
		\SetKwInOut{Input}{input}
		\SetKwInOut{Output}{output}
		\Input{a local independence oracle for $\I^O$}
		\Output{a DG on nodes $O$}
		initialize $\mathcal{D}$ as the complete DG on $O$\;
		\ForEach{$(\alpha,\beta)\in V\times V$}{
			\If{$\alpha \not\rightarrow_{\lambda} \beta\mid 
				\beta$}{delete 
				$\alpha 
				\rightarrow \beta$ from $\mathcal{D}$\; 
			}
		}
		\Return  $\mathcal{D}$
		\newline
		\caption{Trek step}
		\label{subalgo:trekStep}
	\end{subalgorithm}

	\begin{subalgorithm}
		\SetKwInOut{Input}{input}
		\SetKwInOut{Output}{output}
		\Input{a local independence oracle for $\I^O$ and a DG, $\mathcal{D} = 
			(O,E)$}
		\Output{a DG on nodes $O$}
		\ForEach{$(\alpha,\beta)\in V\times V$ such that 
			$\alpha\rightarrow_\mathcal{D}\beta$}{
			\If{\normalfont $\alpha \not\rightarrow_{\lambda} \beta\mid 
				\pa_\mathcal{D}(\beta) \setminus \{\alpha\}$}{delete 
				$\alpha 
				\rightarrow \beta$ from $\mathcal{D}$\; 
			}
		}
		\Return  $\mathcal{D}$
		\newline
		\caption{Parent step}
		\label{subalgo:paStep}
	\end{subalgorithm}

\subsection{ANCESTRY PROPAGATION}

In this section, we describe an additional step which propagates ancestry by 
reusing the output of Subalgorithm \ref{subalgo:trekStep} to remove further 
edges. This comes at a price as one needs faithfulness to ensure soundness. The 
idea is similar to ACI 
\citep{magliacaneACI2016}.

\begin{subalgorithm}
	\SetKwInOut{Input}{input}
	\SetKwInOut{Output}{output}
	\Input{a DG, $\mathcal{D} = 
		(O,E)$}
	\Output{a DG on nodes $O$}
	initialize $E_r = \emptyset$ as the empty edge set\;
	\ForEach{$(\alpha,\beta, \gamma)\in V\times V \times V$ such that
		$\alpha,\beta,\gamma$ are all distinct}{
		\If{$\alpha \rightarrow_\mathcal{D} \beta$, 
			$\beta\not\rightarrow_\mathcal{D}\alpha$, $\beta 
			\rightarrow_\mathcal{D} 
			\gamma$, and 
			$\alpha 
			\not\rightarrow_\mathcal{D} \gamma$}{update $E_r = E_r \cup \{\beta 
			\rightarrow 
			\gamma \}$\; 
		}
	}
	Update $\mathcal{D} = (V, E\setminus E_r)$\;
	\Return  $\mathcal{D}$
	\newline
	\caption{Ancestry propagation}
	\label{subalgo:ancProp1}
\end{subalgorithm}

In ancestry propagation, we exploit the fact that any trek between $\alpha$ and 
$\beta$ (such that $\gamma$ is not on this trek) composed with the edge 
$\beta\rightarrow\gamma$ gives a directed trek 
from $\alpha$ to $\gamma$. We only use the trek 
between $\alpha$ and $\beta$ `in one direction', as a directed trek from 
$\alpha$ to $\beta$. In Subalgorithm 
\ref{subalgo:ancProp2}
(supplementary material), we use a trek 
between $\alpha$ and $\beta$ twice when possible, at the cost of an 
additional 
test.

We can construct an algorithm by first running Subalgorithm 
\ref{subalgo:trekStep}, then Subalgorithm \ref{subalgo:ancProp1}, and finally 
Subalgorithm \ref{subalgo:paStep} (using the output of one subalgorithm as 
input to the next). We will call this the CSAPC algorithm. If we use 
Subalgorithm 
\ref{subalgo:ancProp2} (in the supplementary material) instead of Subalgorithm 
\ref{subalgo:ancProp1}, we will 
call this the CSAP.

\begin{prop}
	If $\I$ is faithful with respect to $\mathcal{D}_0$, then CSAP and 
	CSAPC both
	provide sound 
	learning.
	\label{prop:soundAncProp}
\end{prop}

\begin{figure*}
	\centering
	\begin{subfigure}[t]{.48\textwidth}
		\centering
		\includegraphics[width=1\linewidth, trim = 0 0 0 0]{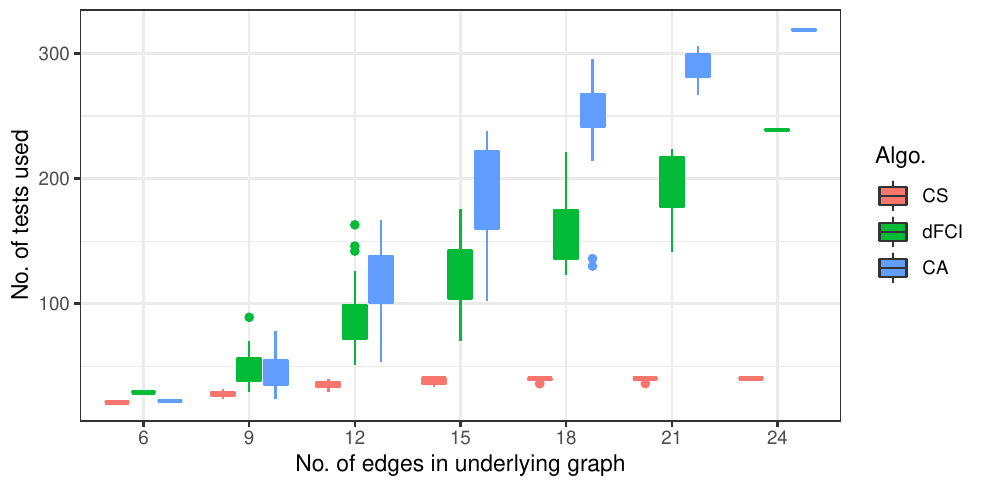}
		\caption{Comparison of number of tests used. For each level of sparsity 
		(number of edges in 
			true graph), we generated 500 graphs, all on 
			5 
			nodes. The number of tests required quickly rises for dFCI and CA 
			while CS spends no more than $2\cdot 5(5 -1)$ tests. The output of 
			dFCI and CA is not considerably more informative as measured by the 
			mean number of excess edges: CS 0.96, dFCI 
			0.07, CA 
			0.81 (average over all levels of sparsity).}
		\label{fig:sub1}
	\end{subfigure} \hfill
	\begin{subfigure}[t]{.48\textwidth}
		\centering
		\includegraphics[width=1\linewidth]{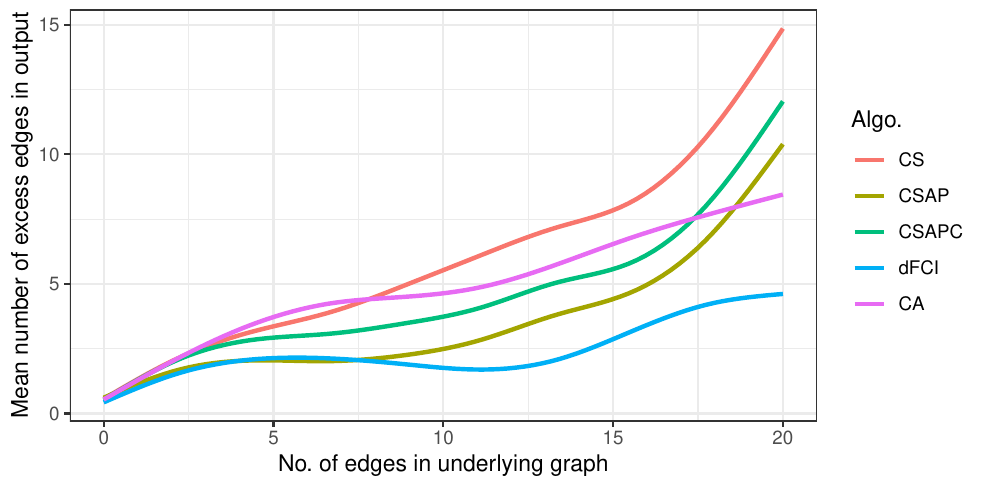}
		\caption{Mean number of excess edges in output graphs 
			for 
			varying 
			numbers of edges (bidirected and directed) in the true graph (all 
			graphs are on 10 nodes), not 
			counting loops.}
		\label{fig:sub2}
		\vfill
	\end{subfigure}
	\caption{Comparison of performance.}
	\label{fig:test}
\end{figure*}

\section{APPLICATION AND SIMULATIONS} 

When evaluating the performance of a sound screening algoritm, the output graph 
is guaranteed to be a supergraph of the true parent graph, and we will say that 
edges 
that are in the output but not in the true graph are {\it excess edges}. For a 
node in a directed graph, the {\it indegree} is the number of directed 
edges adjacent with and pointed into the node, and the {\it outdegree} is the 
number of directed edges adjacent with and pointed away from the node.

One should note that all our experiments are done using an {\it oracle test}, 
i.e., instead of using real or synthetic data, the algorithms simply query an 
oracle for each local independence and receive the correct answer. This tests 
whether or not an algorithm can give good results using an efficient testing 
strategy (i.e., a low number of queries to the oracle) and therefore it 
evaluates the algorithms. This approach separates the 
algorithm from the specific test of local independence and evaluates only the 
algorithm. As such this is highly unrealistic as we would never have access to 
an oracle with real data, however, we should think of these experiments as a 
study of efficiency. The oracle approach to evaluating graphical learning 
algorithms is common in the DAG-based case, see \cite{spirtes2010} for an 
overview.

Also note that the comparison is only made with other constraint-based learning 
algorithms that can actually solve the problem at hand. Learning 
methods that assume full observation (such as the Hawkes methods mentioned in 
Section 
\ref{sec:hawkes}) 
would generally not output a graph with the correct interpretation even in the 
oracle case (see the example in 
Figure \ref{fig:causal}).

\subsection{C. ELEGANS NEURONAL NETWORK}

Caenorhabditis elegans is a roundworm in which the network between neurons has 
been mapped completely \citep{varshneyCelegans2011}. We apply our methods to 
this network as an 
application to a highly complex network. It consists of 279 neurons which are 
connected by both non-directional {\it gap junctions} and directional chemical 
synapses. We will represent the former as an unobserved process and 
the latter as a direct influence which is consistent with the biological system 
\citep{varshneyCelegans2011}. From this network, we sampled subnetworks of 
75 neurons each (details in the supplementary material) and computed the output 
of 
the CS algorithm. These subsampled networks had on average 1109 edges 
(including bidirected edges representing unobserved processes, see the 
supplementary 
material) and on average 424 directed edges. The output graphs had on average 
438 excess edges which is explained by the fact that there are many unobserved 
nodes in the graphs. To compare the output to the true parent graph, we 
computed the rank correlation between the indegrees of the nodes in the output 
graph and the indegrees of the nodes in the true parent graph, and similarly 
for the outdegree 
(indegree 
correlation: 0.94, outdegree correlation: 0.52). 
Finally, we investigated the method's ability to identify the observed nodes of 
highest directed connectivity (i.e., highest in- and outdegrees). The neuronal 
network of c. elegans is inhomogeneous in the sense that some neurons are 
extremely highly connected while others are only very sparsely connected. 
We considered the 
15 nodes of highest indegree/outdegree (out of the 75 observed nodes). On 
average, the CS algorithm placed 
13.4 (in) and 9.2 (out) of these 15 among the 15 most connected nodes.

From the output of the CS algorithm, we can find areas of the neuronal 
network which mediates information from one area to another, e.g., using 
Corollary 
\ref{cor:soundPath}.

\subsection{COMPARISON OF ALGORITHMS}

In this section we compare the proposed causal screening algorithms with 
previously published algorithms that solve similar problems.
\cite{mogensenUAI2018} propose two algorithms, one of which is 
sure 
to output the correct graph when an oracle test is available. They note 
that this complete algorithm is 
computationally very 
expensive and adds little extra information, and therefore we will only 
consider their other algorithm for comparison. We will call this algorithm {\it 
	dynamical} 
FCI (dFCI) as it resembles FCI 
\citep{mogensenUAI2018}. dFCI actually solves a harder learning 
problem (see details in the supplementary material), 
however, it is computationally infeasible for many problems.

The Causal Analysis (CA) algorithm of \cite{meek2014} is a simple screening 
algorithm and we have in this paper argued that it is sound for learning the 
parent graph under the weaker assumption of ancestral 
faithfulness. Even though this algorithm uses a large number of 
tests, it is not guaranteed to provide complete learning 
as 
there may be inseparable nodes that are not adjacent \citep{mogensenUAI2018, 
mogensen2018}.

For the comparison of these algorithms, two aspects are important. As they are 
all sound, one 
aspect is the number of excess edges. The other aspect is 
of course the number of tests needed. The CS and CSAPC 
algorithms use at most $2n(n-1)$ tests and empirically the CSAP uses roughly 
the same number as the two former. This makes them feasible in large 
graphs. The quality of their output is dependent on the sparsity of the true 
graph, 
though the CSAP and CSAPC algorithms can deal considerably better with less 
sparse graphs (Subfigure \ref{fig:sub2}).

\section{DISCUSSION}

We suggested inexpensive constraint-based methods for learning causal structure 
based on testing local independence. An important observation is that local 
independence is asymmetric while conditional independence is symmetric. In a 
certain sense, this may help when constructing learning algorithms as there is 
no need of something like an `orientation phase' as in the FCI. This 
facilitates using very simple methods to give sound causal learning as we do 
not need the independence structure in full to give interesting output. Simple 
screening algorithms may be either adaptive or nonadaptive. We note that 
nonadaptive 
algorithms may be more robust to false conclusions from statistical tests of 
local independence.

The amount of information in the output of the screening 
algorithms 
depends on the sparsity of 
the true graph. However, even in examples with very little sparsity interesting 
structural information can be learned.

We showed that the proposed algorithms have a computational advantage 
over previously published algorithms within this framework. This makes it 
feasible to consider 
causal learning in large networks with unobserved processes. 
We obtained this gain in efficiency in part by outputting only the directed 
part of the causal structure. This means that we may 
be able to answer structural questions, but not questions relating 
to causal effect estimation.

\subsubsection*{Acknowledgments}

This work was supported by VILLUM FONDEN (research grant 13358). We thank Niels 
Richard Hansen and the anonymous reviewers for their helpful comments that
improved this paper.


\bibliography{C:/Users/swmo/Desktop/-/UCPH/PhD/learning/parental/causalScreening/parentalLearning,C:/Users/swmo/Desktop/-/UCPH/PhD/markov/UAIpaperRef}

\newpage

{\centering \huge Supplementary material}

This supplementary material contains additional graph theory, results, and 
definitions, 
as 
well as the 
proofs of the 
main paper.

\section{GRAPH THEORY}

In the main paper, we introduce the class of DGs to represent causal 
structures. One can represent marginalized DGs using the larger class of 
DMGs.
A {\it 
	directed mixed graph} (DMG) is a graph such that any pair of nodes 
$\alpha,\beta\in V$ is joined by a subset of the edges 
$\{\alpha\rightarrow\beta,\alpha\leftarrow\beta,\alpha\leftrightarrow\beta\}$.

We say that edges $\alpha\rightarrow\beta$ and 
$\alpha\leftarrow\beta$ are {\it directed}, and that 
$\alpha\leftrightarrow\beta$ is {\it bidirected}. We say that the edge $\alpha 
\rightarrow \beta$ has a {\it head} at $\beta$ and a {\it tail} at $\alpha$. 
$\alpha\leftrightarrow\beta$ 
has heads at both $\alpha$ and $\beta$. We also introduced a walk 
$\langle \alpha_1, 
e_1, 
\alpha_2, \ldots ,\alpha_n, e_n, \alpha_{n+1} \rangle$. We say 
that $\alpha_1$ and $\alpha_{n+1}$ are endpoint nodes. A nonendpoint node 
$\alpha_i$ on a walk is a {\it collider} if $e_{i-1}$ and $e_i$ both have 
heads 
at $\alpha_i$, and otherwise it is a {\it noncollider}. A cycle is a path 
$\langle\alpha, 
e_1, \ldots, \beta\rangle$ composed with an edge between $\alpha$ and 
$\beta$. 
We say that $\alpha$ is an {\it ancestor} of 
$\beta$ if 
there exists a directed path from $\alpha$ to $\beta$. We 
let $\an(\beta)$ 
denote the set of nodes that are ancestors of $\beta$. For a node set $C$, we 
let $\an(C) = \cup_{\beta\in C} \an(\beta)$. By convention, we say that a 
trivial path (i.e., with no edges) is directed and this means that $C\subseteq 
\an(C)$. 

For DAGs $d$-separation is often used for encoding 
independences. We use the analogous notion of $\mu$-separation which is a 
generalization of 
$\delta$-separation \cite{didelez2000, didelez2008, meek2014, mogensen2018}.

We use the class of DGs to represent the underlying, data-generating 
structure. When only parts of the 
causal system is 
observed, the class of DMGs can be used to represent marginalized DGs 
\cite{mogensen2018}. This can be done using {\it latent projection} 
\cite{vermaEquiAndSynthesis, mogensen2018} which is a map that for a DG (or 
more generally, for a DMG), 
$\mathcal{D} = (V,E)$, and a subset of observed 
nodes/processes, $O\subseteq V$, provides a DMG,
$m(\mathcal{D}, O)$, such that for all $A,B,C \subseteq O$, 

\[
\musepG{A}{B}{C}{\mathcal{D}} \Leftrightarrow 
\musepG{A}{B}{C}{m(\mathcal{D}, 
	O)}.
\]

See 
\cite{mogensen2018} for details on this graphical 
marginalization. 
We say that two DMGs, $\G_1 = (V,E_1), \G_2 = (V,E_2)$, are {\it Markov 
	equivalent} if 

\[
\musepG{A}{B}{C}{\G_1} \Leftrightarrow \musepG{A}{B}{C}{\G_2},
\]

for all $A,B,C \subseteq V$, and we let $[\G_1]$ denote the Markov 
equivalence 
class of $\G_1$. Every Markov 
equivalence class of DMGs has a unique {\it maximal element} 
\cite{mogensen2018}, i.e.,
there exists $\G 
\in [\G_1]$ such that $\G$ is a supergraph of all other graphs in $[\G_1]$. 

For a DMG, $\G$, we will let $D(\G)$ denote the {\it directed part} of 
$\G$, 
i.e., the DG obtained by deleting all bidirected edges from $\G$.

\begin{prop}
	Let $\mathcal{D} = (V,E)$ be a DG, and let $O\subseteq V$. Consider $\G 
	= 
	m(\mathcal{D}, O)$. For $\alpha,\beta\in O$ it holds that $\alpha\in 
	\an_\mathcal{D}(\beta)$ if and only if $\alpha\in \an_{D(\G)}(\beta)$. 
	Furthermore, the directed part of $\G$ equals the parent graph of 
	$\mathcal{D}$ on nodes $O$, i.e., $D(\G) = \mathcal{P}_O(\mathcal{D})$.
\end{prop} 

\begin{proof}
	Note first that $\alpha\in 
	\an_\mathcal{D}(\beta)$ if and only if $\alpha\in 
	\an_\G(\beta)$ \cite{mogensen2018}. Ancestry is only defined by 
	the directed 
	edges, and it follows that $\alpha\in 
	\an_\G(\beta)$ if and only if $\alpha\in 
	\an_{D(\G)}(\beta)$. For the second statement, the definition of the 
	latent 
	projection gives that there is a directed edge from $\alpha$ to $\beta$ 
	in 
	$\G$ if and only if there is a directed path from $\alpha$ to $\beta$ 
	in 
	$\mathcal{D}$ such that no nonendpoint node is in $O$. By definition, 
	this 
	is the parent graph, $\mathcal{P}_O(\mathcal{D})$.
\end{proof}

In words, the above proposition says that if $\G$ is a marginalization 
(done by 
latent projection) of $\mathcal{D}$, then 
the ancestor relations of $\mathcal{D}$ and $D(\G)$ are the 
same among the observed nodes. It also says that our learning target, the 
parent graph, is actually the directed part of the latent projection on the 
observed nodes. In the next subsection, we use this to describe what is 
actually identifiable from the induced independence model of a graph.

\subsection{MAXIMAL GRAPHS AND PARENT GRAPHS}

Under faithfulness of the local independence model and the causal 
graph, we know that the maximal DMG is a correct representation of the 
local 
independence structure in the sense that it encodes exactly the local 
independences that hold in the local independence model. From the maximal 
DMG, 
one can use results on equivalence classes of DMGs to obtain every other 
DMG 
which encodes the observed local independences \citep{mogensen2018} and from 
this graph one can find the parent graph as simply the directed part. 
However, 
it may require an infeasible number of tests to output such a maximal DMG. 
This 
is not 
surprising, seeing that the learning target encodes this complete 
information 
on local independences.

Assume that $\mathcal{D}_0 = (V,E)$ is the underlying causal graph
and that $\G_0 = (O,F), O\subseteq V$ is the 
marginalized graph over 
the observed variables, i.e., the latent projection of $\mathcal{D}_0$. In 
principle, we would like to output $\mathcal{P}(\mathcal{D}_0) = D(\G_0)$, 
the 
directed part
of $\G_0$. However, no algorithm can in general output this graph by 
testing 
only local 
independences as Markov equivalent DMGs may not have the same parent graph. 
Within each Markov equivalence class of DMGs, there is a unique 
maximal graph. Let $\bar{\G}$ denote the maximal graph which is Markov 
equivalent of $\G_0$. The DG $D(\bar{\G})$ is a supergraph of $D(\G_0)$ and 
we 
will say that a learning algorithm is complete if it is guaranteed to 
output 
$D(\bar{\G})$ as no algorithm testing local independence only can identify 
anything more than the equivalence class.

\section{COMPLETE LEARNING}

The CS algorithm provides sound learning of the 
parent graph of a general DMG under the assumption of ancestral 
faithfulness. 
For a subclass of DMGs, the algorithm actually provides complete learning. 
It 
is of interest to find sufficient 
graphical conditions to ensure that the algorithm removes an edge 
$\alpha\rightarrow\beta$ which is not in the true parent graph. In this 
section, we state and prove one such condition which can be 
understood as `the true parent set is always found for unconfounded 
processes'. 
We let 
$\mathcal{D}$ denote the output of 
the CS algorithm.

\begin{prop}
	If $\alpha \not\rightarrow_{\G_0} \beta$ and there is no $\gamma\in 
	V\setminus\{\beta\}$ such that $\gamma\leftrightarrow_{\G_0}\beta$, 
	then 
	$\alpha\not\rightarrow_\mathcal{D}\beta$. 
\end{prop}

\begin{proof}
	Let $\mathcal{D}_1, \mathcal{D}_2, \ldots, \mathcal{D}_N$ denote the 
	DGs 
	that are constructed when running the algorithm by sequentially 
	removing 
	edges, starting from the complete DG, $\mathcal{D}_1$. Consider a 
	connecting walk 
	from 
	$\alpha$ to $\beta$ in $\G_0$. It must be of the form $\alpha \sim 
	\ldots 
	\sim \gamma 
	\rightarrow \beta$, $\gamma\neq \alpha$. Under ancestral faithfulness, 
	the edge $\gamma \rightarrow\beta$ is in $\mathcal{D}$, 
	thus $\gamma\in\pa_{\mathcal{D}_i}(\beta)$ for all $\mathcal{D}_i$ that 
	occur during the algorithm, and therefore when $\langle \alpha, \beta 
	\mid	
	\pa_{\mathcal{D}_i}(\beta) \setminus \{\alpha\} \rangle$ is tested, the 
	walk is closed. Any walk 
	from $\alpha$ to $\beta$ is of this form, thus also closed, and we have 
	that $\musep{\alpha}{\beta}{\pa_{\mathcal{D}_i}(\beta)}$ and therefore 
	$\langle \alpha,\beta \mid \pa_{\mathcal{D}_i}(\beta) \setminus 
	\{\alpha\}\rangle \in \I$. The 
	edge $\alpha 
	\rightarrow_{\mathcal{D}_i} \beta$ is removed and thus absent in the 
	output 
	graph, $\mathcal{D}$.
\end{proof}

\section{ANCESTRY PROPAGATION}

We state Subalgorithm \ref{subalgo:ancProp2} here.

\setcounter{algocf}{3}

\begin{subalgorithm}
	\SetKwInOut{Input}{input}
	\SetKwInOut{Output}{output}
	\Input{a local independence oracle for $\I^O$ and a DG, $\mathcal{D} = 
		(O,E)$}
	\Output{a DG on nodes $O$}
	initialize $E_r = \emptyset$ as the empty edge set\;
	\ForEach{$(\alpha,\beta, \gamma)\in V\times V \times V$ such that
		$\alpha,\beta,\gamma$ are all distinct}{
		\If{$\alpha \sim_\mathcal{D} \beta$, $\beta \rightarrow_\mathcal{D} 
			\gamma$, and 
			$\alpha 
			\not\rightarrow_\mathcal{D} \gamma$}{
			\If{$\langle \alpha,\gamma\mid \emptyset\rangle \in 
				\I^O$}{update 
				$E_r = E_r \cup \{\beta 
				\rightarrow 
				\gamma \}$\;} 
		}	
	}
	Update $\mathcal{D} = (V, E\setminus E_r)$\;
	\Return  $\mathcal{D}$
	\newline
	\caption{Ancestry propagation}
	\label{subalgo:ancProp2}
\end{subalgorithm}

Composing Subalgorithm \ref{subalgo:trekStep}, Subalgorithm 
\ref{subalgo:ancProp2}, and 
Subalgorithm \ref{subalgo:paStep} is referred to as the causal screening, 
ancestry propagation (CSAP) algorithm. If we use Subalgorithm 
\ref{subalgo:ancProp1} instead of Subalgorithm \ref{subalgo:ancProp2}, we 
call 
it the CSAPC algorithm (C for cheap as this does not entail any additional 
independence tests compared to CS).

\section{APPLICATION AND SIMULATIONS}

In this section, we provide some additional details about the c. elegans 
neuronal network and the simulations.

\subsection{C. ELEGANS NEURONAL NETWORK}

For each connection between two neurons a different number of synapses are 
present (ranging from 1 to 37). We only consider connections with more than 
4 
synapses when we define the true underlying network. When sampling the 
subnetworks, highly connected neurons were sampled with 
higher probability to avoid a fully connected subnetwork when marginalizing.

\subsection{COMPARISON OF ALGORITHMS}

As noted in the main paper, the dFCI algorithm solves a strictly harder 
problem. By using the additional graph theory in the supplementary 
material, we 
can understand the output of the dFCI algorithm as a supergraph of the 
maximal 
DMG, $\bar{\G}$. There is 
also a version of the dFCI which is guaranteed to output not only a 
supergraph 
of $\bar{\G}$, but the graph $\bar{\G}$ itself. Clearly, from the output of 
the 
dFCI algorithm, one can simply take the directed part of the output and 
this is 
a supergraph of the underlying parent graph.

\section{PROOFS}

In this section, we provide the proofs of the result in the main paper.

\begin{proof}[Proof of Proposition \ref{prop:pairCau}]
	Let $\mathcal{D}$ denote the causal graph. Assume first that 
	$\alpha\not\rightarrow_\mathcal{D}\beta$. Then $g^{\beta\alpha}$ is 
	identically zero over the observation interval, and it follows directly 
	from the 
	functional form of $\lambda_t^\beta$ that 
	$\alpha\not\rightarrow\beta\mid 
	V\setminus \{\alpha\}$. This shows that the local independence model 
	satisfies the pairwise Markov property with respect to $\mathcal{D}$.
	
	If instead $g^{\beta\alpha} \neq 0$ over $J$, there exists $r\in J$ 
	such 
	that $g^{\beta\alpha}(r) \neq 0$. From continuity of $g^{\beta\alpha}$ 
	there exists a compact interval of positive measure, $I\subseteq 
	J$, such that $\inf_{s\in I}(g^{\beta\alpha}(s)) \geq 
	g_{\min}^{\beta\alpha}$ 
	and $g_{\min}^{\beta\alpha} 
	> 0$. Let $i_0$ and $i_1$ denote the endpoints of this interval, $i_0 < 
	i_1$. We consider now the events 
	
	\begin{align*}
	D_k = (N_
	{T - i_0}^\alpha - N_
	{T - i_1}^\alpha = k, N_T^\gamma = 0 \text{ for all } \gamma \in 
	V\setminus \{\alpha\})
	\label{eq:setsDk}
	\end{align*}

	\noindent $k\in \mathbb{N}_0$. Then under Assumption \ref{ass:gFunc}, 
	for 
	all $k$
	
	\[
	\lambda_T^\beta \mathds{1}_{D_k} \geq \mathds{1}_{D_k} 
	\int_I 
	g^{\beta\alpha}(T-s)\ \md 
	N_s^\alpha \geq
	g_{\min}^{\beta\alpha}\cdot k \cdot \mathds{1}_{D_k}.
	\]
	
	\noindent Assume for contradiction that $\beta$ is locally independent 
	of 
	$\alpha$ given $V\setminus \{\alpha\}$. Then $\lambda_T^\beta = 
	\E(\lambda_T^\beta \mid 
	\mathcal{F}_T^{V}) = \E(\lambda_T^\beta \mid 
	\mathcal{F}_T^{V\setminus 
		\{\alpha\}})$ is constant on $\cup_k D_k$ and furthermore $\PR(D_k) 
	> 
	0$ for 
	all $k$. However, this contradicts the above inequality when $k 
	\rightarrow 
	\infty$.
\end{proof}

\begin{proof}[Proof of Proposition \ref{prop:soundScreen}]
	Let $\mathcal{D}$ denote the DG which is output by the algorithm. We 
	should 
	then show that $\mathcal{P}(\mathcal{D}_0) \subseteq \mathcal{D}$. 
	Assume 
	that $\alpha 
	\rightarrow_{\mathcal{P}(\mathcal{D}_0)} \beta$. In this case, there is 
	a 
	directed path from $\alpha$ to $\beta$ in $\mathcal{D}_0$ such that no 
	nonendpoint node 
	on this directed walk is in $O$ (the 
	observed coordinates). Therefore for any $C\subseteq 
	O\setminus 
	\{\alpha\}$ there exists a directed $\mu$-connecting walk from $\alpha$ 
	to 
	$\beta$ in $\mathcal{D}_0$ and by ancestral faithfulness it follows 
	that 
	$\ind{\alpha}{\beta}{C} \notin \I$. The algorithm starts from the 
	complete
	directed graph, and the above means that the directed edge from 
	$\alpha$ 
	to $\beta$ will not be 
	removed.
\end{proof}

\begin{proof}[Proof of Corollary \ref{cor:soundPath}]
	Consider some directed path from $\alpha$ to $\beta$ in $\mathcal{D}_0$ 
	on 
	which no node is in $C$. Then there is also a directed path from 
	$\alpha$ 
	to $\beta$ on which no nodes is in $C$ in the graph 
	$\mathcal{P}(\mathcal{D}_0)$, and therefore also in the output graph 
	using 
	Proposition \ref{prop:soundScreen}.
\end{proof}

\begin{proof}[Proof of Proposition \ref{prop:trekSep}]
	Assume that there is a $\mu$-connecting walk from $\alpha$ to $\beta$ 
	given 
	$\{\beta\}$. If this walk has no colliders, then it is a directed trek, 
	or 
	can be reduced to one. 
	Otherwise, assume that $\gamma$ is the collider which is the closest to 
	the 
	endpoint
	$\alpha$. Then $\gamma \in \an(\beta)$, and composing the subwalk from 
	$\alpha$ to $\gamma$ with the directed path from $\gamma$ to $\beta$ 
	gives 
	a directed trek, or it can be reduced to one. On the other hand, assume 
	there is a directed trek 
	from 
	$\alpha$ to 
	$\beta$. This is $\mu$-connecting from $\alpha$ to $\beta$ given 
	$\{\beta\}$.
\end{proof}

\begin{proof}[Proof of Proposition \ref{prop:soundAncProp}]
	Assume $\beta\rightarrow_{\mathcal{P}(\mathcal{D}_0)} \gamma$. 
	Subalgorithms 
	\ref{subalgo:trekStep} and \ref{subalgo:paStep} are both
	simple screening algorithms, and they will not remove this edge. Assume 
	for 
	contradiction that 
	$\beta\rightarrow\gamma$ is removed by Subalgorithm 
	\ref{subalgo:ancProp1}. 
	Then 
	there 
	must exist $\alpha\neq \beta,\gamma$ and a directed trek from $\alpha$ 
	to 
	$\beta$ in $\mathcal{D}_0$. On this directed trek, $\gamma$ does not 
	occur 
	as this would imply 
	a directed trek either from $\alpha$ to $\gamma$ or from $\beta$ to 
	$\alpha$, thus implying $\alpha\rightarrow_\mathcal{D}\gamma$ or 
	$\beta\rightarrow_\mathcal{D}\alpha$, respectively ($\mathcal{D}$ is 
	the 
	output graph of Subalgorithm \ref{subalgo:trekStep}). As $\gamma$ does not 
	occur on the trek, composing this trek with the edge $\beta 
	\rightarrow \gamma$ would give a directed trek from $\alpha$ to 
	$\gamma$. 
	By faithfulness, $\langle \alpha,\gamma\mid \gamma\rangle \notin \I$, 
	and 
	this is a contradiction as $\alpha\rightarrow\gamma$ would not have 
	been 
	removed during Subalgorithm 1.
	
	We consider instead CSAP. Assume for contradiction that 
	$\beta\rightarrow\gamma$ is removed during Subalgorithm 
	\ref{subalgo:ancProp2}. There exists in $\mathcal{D}_0$ either a 
	directed 
	trek from $\alpha$ 
	to $\beta$ or a directed trek from $\beta$ to $\alpha$. If $\gamma$ is 
	on 
	this trek, then $\gamma$ is not $\mu$-separated from $\alpha$ given the 
	empty set (recall that there are loops at all nodes, therefore also at 
	$\gamma$), and using faithfulness we conclude that $\gamma$ is not on 
	this 
	trek. Composing it with the edge $\beta\rightarrow\gamma$ would give a 
	directed trek from $\alpha$ to $\gamma$ and using faithfulness we 
	obtain a 
	contradiction.	
\end{proof}



\end{document}